\documentclass[11pt]{article}
\usepackage{amsmath}
\usepackage{a4wide}
\usepackage{anysize}
\usepackage{latexsym}
\usepackage{amsthm, amssymb}
\usepackage{epsfig}
\usepackage{url}
\usepackage{pstricks}
\usepackage{tabls}
\usepackage{tabularx}
\usepackage{pst-node}
\usepackage{pst-plot}
\usepackage{appendix}
\usepackage{verbatim}
\usepackage{multirow,colortbl}

\marginsize{2.5cm}{2.5cm}{1.5cm}{2.5cm}

\pslongbox{MyFrame}{\psframebox}
{\MyFrame\begin{minipage}{#1}}%
{\end{minipage}\endMyFrame}

\theoremstyle{plain}
\newtheorem{theorem}{Theorem}
\newtheorem{lemma}[theorem]{Lemma}

\newtheorem{claim}{Claim}

\theoremstyle{definition}

\newtheorem{problem}{Problem}

\theoremstyle{remark}

\newtheorem{fact}{Fact}

\DeclareMathOperator{\suchthat}{\text{ }|\text{ }}
\DeclareMathOperator{\opt}{OPT}

\DeclareMathOperator{\poly}{poly}

\begin{document}
\title{Metrical Service Systems with Multiple Servers}
\author{Ashish Chiplunkar \qquad Sundar Vishwanathan\\
Department of Computer Science and Engineering,\\
Indian Institute of Technology Bombay,\\
Mumbai, India\\
\texttt{\{ashishc, sundar\}@cse.iitb.ac.in}
}
\date{}
\maketitle

\begin{abstract}
We study the problem of metrical service systems with multiple servers (MSSMS), which generalizes two well-known problems -- the $k$-server problem, and metrical service systems. 
The MSSMS problem is to service requests, each of which is an $l$-point subset of a metric space, using $k$ servers, with the objective of minimizing the total distance traveled by the servers.

Feuerstein \cite{Feuerstein98} initiated a study of this problem by proving upper and lower bounds on the deterministic competitive ratio for uniform metric spaces. We improve Feuerstein's analysis of the upper bound and prove that his algorithm achieves a competitive ratio of $k\left({{k+l}\choose{l}}-1\right)$.
In the randomized online setting, for uniform metric spaces, we give an algorithm which achieves a competitive ratio $\mathcal{O}(k^3\log l)$, beating the deterministic lower bound of ${{k+l}\choose{l}}-1$.
We prove that any randomized algorithm for MSSMS on uniform metric spaces must be $\Omega(\log kl)$-competitive. We then prove an improved lower bound of ${{k+2l-1}\choose{k}}-{{k+l-1}\choose{k}}$ on the competitive ratio of any deterministic algorithm for $(k,l)$-MSSMS, on general metric spaces. 
In the offline setting, we give a pseudo-approximation algorithm for $(k,l)$-MSSMS on general metric spaces, which achieves an approximation ratio of $l$ using $kl$ servers. We also prove a matching hardness result, that a pseudo-approximation with less than $kl$ servers is unlikely, even for uniform metric spaces. For general metric spaces, we highlight the limitations of a few popular techniques, that have been used in algorithm design for the $k$-server problem and metrical service systems.
\end{abstract}

\section{Introduction}

The problem of metrical service systems with multiple servers (MSSMS) generalizes two well-known problems -- the $k$-server problem \cite{ManasseMS88} and metrical service systems (MSS) \cite{ChrobakL92,ManasseMS88}. These problems share a common paradigm, that there is an underlying metric space 
and requests are to be served by moving servers on the metric space, in such a way that the total distance traveled by the servers is minimized.
For a problem in the online setting, the input is revealed to an algorithm piece by piece, and the algorithm must take irrevocable decisions on seeing each piece. In case of the aforementioned problems, each piece in the input is a request, and the irrevocable decisions are the movements of the servers. A (possibly randomized) online algorithm is said to be $c$-competitive if, on every input, it returns a solution whose (expected) cost is at most $c$ times the cost of optimal solution for that input. The book by Borodin and El-Yaniv \cite{BorodinE} gives a nice comprehensive introduction to online algorithms and competitive analysis.

\noindent
\textbf{The $k$-server problem:} The $k$-server problem of Manasse et. al. \cite{ManasseMS88} is, arguably, the most famous among the problems that are naturally posed in the online setting. The following quote by Koutsoupias, in his beautiful survey on the $k$-server problem \cite{Koutsoupias09}, upholds the importance of this problem.
\begin{quote}
\textit{The $k$-server problem is perhaps the most influential online problem: natural, crisp, with a surprising technical depth that manifests the richness of competitive analysis. The $k$-server conjecture, which was posed more than two decades ago when the problem was first studied within the competitive analysis framework, is still open and has been a major driving force for the development of the area online algorithms.}
\end{quote}
In the $k$-server problem, we have $k$ servers occupying points in a metric space. Each request is a point in the metric space. To serve the request, one of the servers has to be moved to the requested point. The $k$-server conjecture referred to in the quote states that there is a $k$-competitive deterministic algorithm for the $k$-server problem.

Manasse et. al. \cite{ManasseMS88} proved a lower bound of $k$ on the competitive ratio of any deterministic algorithm on any metric space with more than $k$ points. They proved that the competitive ratio is $k$ only for very specific cases, and posed the $k$-server conjecture.
The conjecture has been shown to hold for a few metric spaces, for example, the line \cite{ChrobakKPV91} and tree metric spaces \cite{ChrobakL91}.
Fiat, Rabani and Ravid \cite{FiatRR90} were the first to give an algorithm for the $k$-server problem, with competitive ratio bounded by a function of $k$, which was later improved by Grove \cite{Grove91,BartalG00}. The breakthrough result was due to Koutsoupias and Papadimitriou, who proved that an algorithm, first proposed by Chrobak and Larmore, and called the \textit{Work Function Algorithm} (WFA), is $(2k-1)$-competitive \cite{KoutsoupiasP95}.
In case of randomized algorithms, the best known lower bound that holds for every metric space is $\Omega(\log k/\log\log k)$ due to Bartal et. al. \cite{BartalBM01}, and there do exist metric spaces with a lower bound more than $H_k$ \cite{KarlinMMO94}. No better algorithm than the Work Function Algorithm is known, even with randomization. The randomized $k$-server conjecture states that there exists a randomized algorithm for the $k$-server problem with competitive ratio $\mathcal{O}(\log k)$ on any metric space. 
Recent developments, which ingeniously adapt the primal-dual framework to the online setting, have been applied to the $k$-server problem, culminating in a $\mathcal{O}(\log k)$-competitive randomized algorithm for star metrics \cite{BansalBN07}, and a $\mathcal{O}(\poly\log(k)\poly\log(n))$-competitive randomized algorithm on metric spaces of $n$ points \cite{BansalBMN11}.

\noindent
\textbf{The Generalized Server Problem:} This is a generalization of the $k$-server problem, in which the metric space is a disjoint union of $k$ metric spaces, mutually separated by an infinite distance. A server is located at one point in each of the subspaces. A request is a set of $k$ points, one from each subspace. The request is to be served by moving some server to the requested point which lies in its subspace. An interesting problem, called the Weighted Server problem \cite{FiatR94} is a particular case of the Generalized Server problem. This problem is same as the $k$-server problem, except that the servers have different weights, and the cost of moving a server is equal to the product of its weight and the distance covered. We can thus think of this as the Generalized Server problem where the metric spaces are scaled copies of one another. Fiat and Ricklin \cite{FiatR94} were the first to study the Weighted Server problem. They gave a deterministic algorithm with a competitive ratio of $2^{2^{\mathcal{O}(k)}}$ for uniform metric spaces. They proved that for every metric space there exist weights so that any deterministic algorithm must have a competitive ratio of $(k+1)!/2$. Chrobak and Sgall \cite{ChrobakS04} studied the weighted $2$-server problem on uniform spaces, and proved that the Work Function Algorithm achieves the best possible competitive ratio of $5$. They proved that in contrast with the $k$-server problem, there does not exist a memoryless randomized algorithm with a finite competitive ratio, even for the weighted $2$-server problem. Recently Sitters \cite{Sitters11} proved that the Work Function Algorithm, in fact, is competitive for the generalized $2$-server problem.

\noindent
\textbf{Metrical Service System:} The term Metrical Service System (MSS) was coined by Chrobak and Larmore \cite{ChrobakL93} for the following problem. We have a single server in an underlying metric space. Each request is a set of $l$ points from the metric space, where $l$ is called the \textit{width}, and is a parameter to the problem. To serve a request, the server has to be dispatched to one of the requested points.

Finding shortest paths is a fundamental problem in the offline setting. In the online setting it is posed as the problem of Layered Graph Traversal (LGT). This problem, introduced by Papadimitriou and Yannakakis \cite{PapadimitriouY91}, was a precursor of MSS. Fiat et. al. \cite{FiatFKRRV98} proved that MSS and LGT are in fact equivalent problems. That is, there is a $c$-competitive algorithm for MSS if and only if there is a $c$-competitive algorithm for LGT. They also proved that there exists a metric space on which the competitive ratio of any deterministic algorithm for the MSS problem is $\Omega(2^l)$.
Further, they gave an $\mathcal{O}(9^l)$-competitive algorithm. They proved that $l/2$ is a lower bound on the competitive ratio of any randomized algorithm for LGT. Ramesh \cite{Ramesh95} gave a better deterministic algorithm for LGT, which achieves a competitive ratio of $l^32^l$, and a randomized $l^{13}$-competitive algorithm. He proved that there exists a metric space on which any randomized algorithm must have competitive ratio $\Omega(l^2/\log^{1+\varepsilon}l)$, for any $\varepsilon>0$. Burley \cite{Burley96} proved that a variant of the Work Function Algorithm is $\mathcal{O}(l\cdot2^l)$-competitive for the MSS (and hence, the LGT) problem. For the uniform metric space, Chrobak and Larmore \cite{ChrobakL93} proved a lower bound of $l$ on the competitive ratio of any deterministic algorithm, and also gave an algorithm achieving this bound. It is easily seen that the lower bound holds for any metric space with at least $l+1$ points.

\noindent
\textbf{Metrical Service System with Multiple Servers:} In a natural generalization of both the $k$-server problem and metrical service system, we have $k$ servers on an underlying metric space, and each request is a set of $l$ points from the metric space. To serve a request, one of the $k$ servers has to move to one of the $l$ requested points. We call this problem Metrical Service System with Multiple Servers, with parameters $k$ and $l$ ($(k,l)$-MSSMS). It is easy to see that this problem is in fact, a further generalization of the Generalized Server problem. Feuerstein \cite{Feuerstein98} studied this problem for uniform metric spaces and called it the Uniform Service System with parameters $k$ and $l$ (USS$(k,l)$). He proved a lower bound of ${{k+l}\choose{k}}-1$ on the competitive ratio of any deterministic algorithm for this problem. In fact, this proof holds for any (not necessarily uniform) metric space with at least $k+l$ points. Feuerstein also gave an algorithm and proved that its competitive ratio is $k\cdot\min\left(\frac{k^{l+1}-k}{k-1},\sum_{i=0}^{k-2}l^i+l^k\right)$. He concluded the paper with the following comment.
\begin{quote}
\textit{An interesting subject of future research is to extend USS$(k,l)$\footnote{Feuerstein used `$w$' for the width parameter, while we use `$l$'.} to non-uniform metric spaces. This would extend both the work in this paper and the work by Chrobak and Larmore \cite{ChrobakL92} on Metrical Service Systems, where only one server is considered.}
\end{quote}





\noindent
\textbf{Our results}: In Section \ref{sec_hitting_set} we present a simple analysis of Feuerstein's algorithm, which improves the bound on its competitive ratio proved in \cite{Feuerstein98}, to $k\cdot\left({{k+l}\choose{l}}-1\right)$. In Section \ref{sec_uniform_random} we give a $\mathcal{O}(k^3\log l)$-competitive randomized algorithm 
on uniform metric spaces, which beats the deterministic lower bound by an exponential factor. We also give an $\Omega(\log kl)$ lower bound on the competitive ratio of any randomized algorithm. 
We then give, in Section \ref{sec_general_lb}, a lower bound of ${{k+2l-1}\choose{k}}-{{k+l-1}\choose{k}}$ on the competitive ratio of any deterministic algorithm for $(k,l)$-MSSMS, even when the metric space has only two distances. For arbitrary metric spaces, we prove a lower bound which is exponential in $k$ for any fixed $l$.
These lower bounds are improvements to the bound of ${{k+l}\choose{k}}-1$ in \cite{Feuerstein98}. In Section \ref{sec_offline} we consider the offline $(k,l)$-MSSMS problem on arbitrary metric spaces, and we give a pseudo-approximation by a factor of $l$ using $kl$ servers. We prove a matching lower bound, assuming the Unique Games Conjecture. We conclude with a number of interesting open problems in Section \ref{sec_open_problems}. 

\section{Uniform metric spaces: The Hitting Set Algorithm}\label{sec_hitting_set}

In this section, 
we analyze the algorithm for $(k,l)$-MSSMS on uniform metric spaces given by Feuerstein (Section 2 of \cite{Feuerstein98}), which he calls the \textit{Hitting Set} algorithm. Feuerstein proved that the competitive ratio of this algorithm is at most $k\cdot\min\left(\frac{k^{l+1}-k}{k-1},\sum_{i=0}^{k-2}l^i+l^k\right)$, whereas we prove an asymptotically better bound,\footnote{For instance, when $k=\Theta(l)$, Feuerstein's bound is $\Omega(l^l)$, whereas ours is $\mathcal{O}(c^l)$ for some constant $c$.} of $k\cdot\left({{k+l}\choose{l}}-1\right)$. 

The Hitting Set (HS) algorithm can be described as follows. HS divides the request sequence into phases, the first phase starting with the first request of the sequence. We say that a request produces a fault whenever the requested set of points is disjoint from the set of points occupied by the servers. 
Each time a request produces a fault, the algorithm behaves as follows. First, it computes a minimum cardinality set $H$ of points that intersects all the requests that produced a fault during the current phase. If $|H|\leq k$ then any $|H|$ servers are made to occupy all points in $H$. Otherwise, if $|H|>k$, then the phase terminates and a new phase begins with the current request.

\begin{theorem}
HS is an $k\cdot\left({{k+l}\choose{l}}-1\right)$-competitive algorithm for $(k,l)$-MSSMS.
\end{theorem}

\begin{proof}
Feuerstein observed that the adversary must incur a cost of at least one per phase. He then proved that at most $\min\left(\frac{k^{l+1}-k}{k-1},\sum_{i=0}^{k-2}l^i+l^k\right)$ requests can produce faults, in any phase. We improve this upper bound to ${{k+l}\choose{l}}-1$, and our claim follows, since the algorithm pays at most $k$, for every request that produces a fault. Our proof uses the following result from extremal combinatorics, due to Lov\'{a}sz \cite{Lovasz_skewBollobas}.

\begin{fact}[Skew Bollob\'{a}s Theorem]
Let $A_1,\ldots,A_r$ and $B_1,\ldots,B_r$ be sets such that $|A_i|=l$ and $|B_i|=k$ for all $i$. Suppose $A_i\cap B_i=\emptyset$ for all $i$, and $A_j\cap B_i\neq\emptyset$ for all $i,j$ with $j<i$. Then $r\leq{{k+l}\choose{k}}$.
\end{fact}

Let $A_1,\ldots, A_{r-1}$ be the requests that produced a fault in a given phase, and let $A_r$ be the first request in the next phase, which must also have produced a fault. Let $B_i$ be the set of points occupied by the servers when the request $A_i$ was given. Clearly, $|A_i|=l$ and $|B_i|=k$ for all $i$.
Since $A_i$ produced a fault, we have for every $i$ with $1\leq i\leq r$, $A_i\cap B_i=\emptyset$. 
On the other hand, since no request between $A_{i-1}$ and $A_i$ produced a fault, the hitting set chosen to serve $A_{i-1}$ must be $B_i$ itself. By the definition of the algorithm, for any $i$, $j$ with $1\leq j<i\leq r$ we have $A_j\cap B_i\neq\emptyset$. 
Applying the skew Bollob\'{a}s theorem, we have $r\leq{{k+l}\choose{k}}$. Thus, the number of requests in the phase, that produced a fault, is $r-1\leq{{k+l}\choose{k}}-1$, as required.
\end{proof}

\section{Uniform metric spaces: Randomized bounds}\label{sec_uniform_random}

In this section, we give a randomized version of the Hitting Set algorithm from Section \ref{sec_hitting_set}, which we call the Randomized Hitting Set (RHS) algorithm, and prove that its competitive ratio is $\mathcal{O}(k^3\log l)$. The algorithm is as follows.

RHS divides the sequence of requests in phases, just like the Hitting Set algorithm. Each time a request produces a fault, the algorithm behaves as follows. First, it computes $s$, the minimum cardinality of a set of points that intersects all the requests in the current phase given so far. If $s\leq k$, it chooses a set $H$ uniformly at random from the collection of all the hitting sets of size $s$, and then any $s$ servers are made to occupy all points in $H$. Otherwise, if $s>k$, the current phase ends and a new phase begins with the current request.

It is easily seen that the adversary must incur a cost of at least one per phase. We prove that the expected cost of RHS is $\mathcal{O}(k^3\log l)$ per phase. Note that the value of $s$ increases from $1$ to $k$ as the algorithm progresses in a phase. We divide the phase into $k$ sub-phases, where the $i^{\text{\tiny{th}}}$ sub-phase is the part of the phase when the value of $s$ is equal to $i$. We require the following combinatorial lemma.

\begin{lemma}\label{lem_hittingset}
Let $s$ be the size of the smallest hitting set of an $l$-uniform set system $\mathcal{S}$. Then the number of minimum hitting sets of $\mathcal{S}$ is at most $l^s$.
\end{lemma}

\begin{proof}
We prove the claim by induction on $s$. For $s=1$, the claim is obvious. For a general $s$, let $\mathcal{H}$ be the collection of all minimum hitting sets of $\mathcal{S}$. Let $A'$ be an arbitrary set in $\mathcal{S}$. For each $e\in A'$, let $\mathcal{S}_e=\{A\in\mathcal{S}\suchthat e\notin A\}$ and $\mathcal{H}_e=\{B\setminus\{e\}\suchthat B\in\mathcal{H}\text{, }e\in B\}$. 
Then each set in $\mathcal{H}_e$ hits each set in $\mathcal{S}_e$, so the size of smallest hitting set of $\mathcal{S}_e$ is at most $s-1$. Further, if $\mathcal{S}_e$ has a hitting set $H$ of size less than $s-1$, then $H\cup\{e\}$ would be a hitting set of $\mathcal{S}$ of size less than $s$. This is a contradiction. Thus, the size of the smallest hitting set of $\mathcal{S}_e$ is exactly $s-1$. Therefore, by induction, we have $|\mathcal{H}_e|\leq l^{s-1}$, for all $e\in A'$. Now each set in $\mathcal{H}$ intersects $A'$ in some $e$, and therefore, contributes to $\mathcal{H}_e$, for some $e$. Thus, $|\mathcal{H}|\leq\sum_{e\in A'}|\mathcal{H}_e|\leq l\cdot l^{s-1}=l^s$.
\end{proof}

\begin{theorem}
RHS is an $\mathcal{O}(k^3\log l)$ competitive algorithm for $(k,l)$-MSSMS.
\end{theorem}

\begin{proof}
First, observe that the boundaries between sub-phases are completely determined by the requests, and are independent of the algorithm. Consider the $s$'th sub-phase of any phase. Let $A_1,\ldots,A_r$ be the sets requested in this sub-phase. Let $\mathcal{S}_0$ be the collection of sets requested in the current phase before the $s$'th sub-phase started (that is, before $A_1$ was requested), and let $\mathcal{S}_i=\mathcal{S}_0\cup\{A_1,\ldots,A_i\}$. For $i\geq1$, let $\mathcal{H}_i$ be the collection of all minimum hitting sets of $\mathcal{S}_i$, and $h_i=|\mathcal{H}_i|$. Thus, each $\mathcal{H}_i$ is an $s$-uniform set system, and $\mathcal{H}_1\supseteq\cdots\supseteq\mathcal{H}_r\neq\emptyset$.

Let $T$ be the set of all $i$ such that the algorithm faulted on $A_i$, and let $H_i$ be the hitting set chosen uniformly at random from $\mathcal{H}_i$, to serve $A_i$. Clearly, $1\in T$, otherwise, the $s$'th sub-phase would not have started with $A_1$. We need an upper bound on $\mathsf{E}[|T|]=\sum_{i=1}^r\Pr[i\in T]$. 

For $i\geq2$ and $j<i$, say that the event $E_{ji}$ has occurred if $i,j\in T$, but no $i'$ between $j$ and $i$ is in $T$. Given that $j\in T$, this event occurs exactly when $H_j$, chosen uniformly at random from $\mathcal{H}_j$, turns out to be in $\mathcal{H}_{i-1}$, but not in $\mathcal{H}_i$. Thus, we have $\Pr[E_{ji}\suchthat j\in T]=(h_{i-1}-h_i)/h_j$. Note that for a fixed $i$, the events $E_{ji}$ are pairwise disjoint, and $i\in T$ if and only if $E_{ji}$ occurs for some $j<i$. Hence,
\begin{eqnarray*}
\Pr[i\in T] & = & \sum_{j=1}^{i-1}\Pr[E_{ji}]=\sum_{j=1}^{i-1}\Pr[E_{ji}\suchthat j\in T]\Pr[j\in T]\\
 & = & \sum_{j=1}^{i-1}\frac{h_{i-1}-h_i}{h_j}\times\Pr[j\in T]=(h_{i-1}-h_i)\sum_{j=1}^{i-1}\frac{\Pr[j\in T]}{h_j}
\end{eqnarray*}

We will now inductively prove that for $i\geq2$, $\Pr[i\in T]=(h_{i-1}-h_i)/h_{i-1}$. This is evident for $i=2$, since $\Pr[1\in T]=1$. For an arbitrary $i>2$, we have
\begin{eqnarray*}
\Pr[i\in T] & = & (h_{i-1}-h_i)\sum_{j=1}^{i-1}\frac{\Pr[j\in T]}{h_j}=(h_{i-1}-h_i)\left[\frac{1}{h_1}+\sum_{j=2}^{i-1}\frac{h_{j-1}-h_j}{h_{j-1}h_j}\right]\\
 & = & (h_{i-1}-h_i)\left[\frac{1}{h_1}+\sum_{j=2}^{i-1}\left(\frac{1}{h_j}-\frac{1}{h_{j-1}}\right)\right]=\frac{h_{i-1}-h_i}{h_{i-1}}
\end{eqnarray*}

The expected size of $T$ is given by
\[\sum_{i=1}^r\Pr[i\in T]=1+\sum_{i=2}^r\frac{h_{i-1}-h_i}{h_{i-1}}\leq\sum_{j=1}^{h_1}\frac{1}{j}\]
Now, $h_1=|\mathcal{H}_1|$ and $\mathcal{H}_1$ is the collection of all minimum hitting sets of the $l$-uniform set system $\mathcal{S}_1$, with minimum hitting set size $s$. By Lemma \ref{lem_hittingset}, $h_1\leq l^s$. Hence, $\mathsf{E}[|T|]$ is $\mathcal{O}(s\log l)$. The cost incurred for every request which produced a fault is at most $s$, and hence, the total cost is $\mathcal{O}(s^2\log l)$ in the $s$'th sub-phase. Summing over $s$ from $1$ to $k$, we infer that the cost incurred in an entire phase is $\mathcal{O}(k^3\log l)$.
\end{proof}

While we have a $O(k^3\log l)$-competitive algorithm, we also have the following lower bound on the competitive ratio of any randomized online algorithm for $(k,l)$-MSSMS on uniform metric spaces.

\begin{theorem}\label{thm_rand_lb}
The competitive ratio of any randomized online algorithm for $(k,l)$-MSSMS against an oblivious adversary \footnote{An oblivious adversary is an adversary who does not have access to the random bits used by the algorithm.} is $\Omega(\log kl)$.
\end{theorem}

We use a form of Yao's principle \cite{Yao77} to prove the above theorem.
The required form of this principle is as follows \cite{BorodinE99,StougieV02}.

\begin{fact}[Yao's principle]\label{fact_yao}
Let $\mathfrak{A}$ denote the set of deterministic online algorithms for an online minimization problem. If $P$ is a distribution over input sequences such that for some real number $c\geq1$, $\inf_{\mathcal{A}\in\mathfrak{A}}\mathop{\mathbb{E}}_{\rho\sim P}\left[\mathcal{A}(\rho)\right]\geq c\cdot\mathop{\mathbb{E}}_{\rho\sim P}\left[\opt(\rho)\right]$, 
then $c$ is a lower bound on the competitive ratio of any randomized algorithm against an oblivious adversary.
\end{fact}

Here $\opt(\rho)$ denotes the optimum cost of serving the sequence $\rho$.

\begin{proof}[Proof of Theorem \ref{thm_rand_lb}]
Let $\mathcal{M}$ be the uniform metric space with $k+l$ points, and let $B_0$ be the set of points occupied by the servers initially. Let $\mathcal{S}$ be the collection of all sets of points $B$ such that $|B|=k$ and $|B\setminus B_0|=1$. The distribution $P$, required to use Fact \ref{fact_yao}, is given by the following random process. We generate sets of size $l$ uniformly at random and request them, until the complements of all sets in $\mathcal{S}$ have been picked. We do not request the last such set picked, say $A$. Note that $\overline{A}\in\mathcal{S}$, and $A$ does not appear in the input. Let $\rho$ denote this random input.

The adversary can serve the request sequence $\rho$, simply by covering the points in $\overline{A}$. Since the metric space has only $k+l$ points, $\overline{A}$ intersects with all the sets that appear in $\rho$. Further, since $\overline{A}\in\mathcal{S}$, $|\overline{A}\setminus B_0|=1$. Thus 
$\opt(\rho)=1$.

The algorithm incurs a cost exactly when the request is the complement of the set of positions occupied by its servers. This happens with probability $1/{{k+l}\choose{l}}$, and hence the expected cost incurred by the algorithm for any request is $1/{{k+l}\choose{l}}$. Let $X_i$ be the random variable representing the cost on the $i^{\text{\tiny{th}}}$ request, and $Y$ be the random variable representing the number of requests in $\rho$. Then $\mathop{\mathbb{E}}_{\rho\sim P}\left[X_i|Y\geq i\right]=1/{{k+l}\choose{l}}$ and $\mathop{\mathbb{E}}_{\rho\sim P}\left[X_i|Y<i\right]=0$. Thus, $\mathop{\mathbb{E}}_{\rho\sim P}\left[X_i\right]=1/{{k+l}\choose{l}}\cdot\Pr_{\rho\sim P}\left[Y\geq i\right]$, and therefore, the expected total cost is given by $\sum_{i=1}^{\infty}\mathop{\mathbb{E}}_{\rho\sim P}\left[X_i\right]=\left(\sum_{i=1}^{\infty}\Pr_{\rho\sim P}\left[Y\geq i\right]\right)/{{k+l}\choose{l}}=\mathop{\mathbb{E}}_{\rho\sim P}\left[Y\right]/{{k+l}\choose{l}}$. 
By the coupon collector argument, the expected number of requests in $\rho$ is ${{k+l}\choose{l}}\cdot H(|\mathcal{S}|)-1={{k+l}\choose{l}}\cdot H(kl)-1$. Thus the expected cost incurred by the algorithm is $\left({{k+l}\choose{l}}\cdot H(kl)-1\right)/{{k+l}\choose{l}}=\Omega(\log kl)$.
\end{proof}

\section{General metric spaces: Lower bounds}\label{sec_general_lb}

In this section we prove two lower bounds on the competitive ratio of any algorithm for $(k,l)$-MSSMS. Both of these are strictly better than Feuerstein's lower bound \cite{Feuerstein98} of ${{k+l}\choose{l}}-1$, for all $k>1$, $l>1$. The basic idea of both proofs, is that the adversary runs many algorithms in parallel with the online algorithm, in such a way that the total cost incurred by these algorithms, is at most a constant times the cost incurred by the online algorithm. This idea is due to Manasse et. al. \cite{ManasseMS88}. For the first lower bound, we give a metric space with two distances, on which an adversary can force any algorithm to perform poorly.

\begin{theorem}\label{thm_lower_bound_online}
There is a metric space with two distances, on which the competitive ratio of any deterministic online algorithm for the $(k,l)$-MSSMS problem is more than ${{k+2l-1}\choose{k}}-{{k+l-1}\choose{k}}$.
\end{theorem}

\begin{proof}
The required metric space is given by the set of points $M=[l]\times[k+1]$, where we will call the set of points $\{i\}\times[k+1]$ the $i$'th cluster $M_i$. A pair of distinct points is separated by a distance $1$ if they belong to the same cluster, and a large distance $D$ otherwise. We will determine $D$ later. Assume that initially the servers occupy the points $(1,1),\ldots,(1,k)$.
\footnote{This is for convenience only. The proof works with any choice of initial positions.}

Fix an online algorithm. The adversary forces the algorithm to perform poorly, as follows. At any time, the adversary picks, from each of the $l$ clusters, the lowest numbered point which is not occupied by the algorithm's servers. It gives this set of points as the next request. Note that this is possible because each cluster has $k+1$ points but the algorithm has only $k$ servers.
Since none of the points in this request set are occupied by the algorithm's servers, the algorithm incurs a nonzero cost to serve each request. We may assume without loss of generality that the algorithm is \textit{lazy}, that is, it moves its servers only to serve requests, and thus, moves exactly one server for every request.\footnote{This is a standard argument; one can postpone the movements of a server to the instant when it gets a chance to serve a request, and this does not increase the distance traveled by the server.} For the purpose of analysis, we divide the input into phases, based on the behaviour of the algorithm to serve the input. 
A phase ends whenever the algorithm moves one of its servers from one cluster to another, and a new phase begins from the next request.

By definition of a phase, the number of the algorithm's servers in each of the $l$ clusters does not change during a phase. We say that a phase is of type $(k_1,\ldots,k_l)$ if the algorithm has $k_i$ servers inside the cluster $M_i$, during that phase. Thus, the zeroth phase is of type $(k,0,\ldots,0)$. Suppose the requests are given to the algorithm ad infinitum. We claim that if the algorithm has a finite competitive ratio, there must be infinitely many phases of each type. Suppose this is not the case. Then there exists a phase type $\kappa=(k_1,\ldots,k_l)$ such that there are only finitely many phases of type $\kappa$ and each of these phases is finite. The adversary serves the requests arbitrarily until the last phase of type $\kappa$ is over. Once the last phase of type $\kappa$ is over, the adversary places $k_i$ servers inside $M_i$ at $(i,1),\ldots,(i,k_i)$ for each $i$ permanently. Henceforward, the adversary always has more servers than the online algorithm, in some cluster. It is then able to serve all the future requests at zero cost, while the algorithm pays for each request. Thus the algorithm cannot have a finite competitive ratio. In particular, this claim implies that the algorithm must change phases infinitely often.

The adversary runs many algorithms in parallel with the online algorithm. We will denote this number by $h(k,l)$ and calculate it later. For each phase type $\kappa=(k_1,\ldots,k_l)$, the adversary runs $g(\kappa)=\prod_{i=1}^l(k_i+1)-1$ algorithms, each of which has $k_i$ servers in $M_i$, for each $i$, permanently. These algorithms maintain the following invariant. During a phase of type $\kappa'\neq\kappa$, all these $g(\kappa)$ algorithms place their $k_i$ servers on the points $(i,1),\ldots,(i,k_i)\in M_i$, for each $i$. As before, these algorithms have more servers than the online algorithm, in some cluster. Thus they serve the requests at zero cost.

On the other hand, during a phase of type $\kappa$, consider the sets of points $M'_i=\{(i,j)\suchthat j\leq k_i+1\}$, for $i=1$ to $l$. There are $\prod_{i=1}^l(k_i+1)$ ways of placing $k_1$ servers at distinct points in $M'_1$, $k_2$ servers at distinct points in $M'_2$, $\ldots$, $k_l$ servers at distinct points in $M'_l$. Since the request sets in the phase are subsets of $\bigcup_{i=1}^lM'_i$, the set of points $S$ occupied by the online algorithm's servers corresponds to one of the $\prod_{i=1}^l(k_i+1)$ ways. The adversary maintains the invariant, that the configurations of its $\prod_{i=1}^l(k_i+1)-1$ algorithms of type $\kappa$ correspond to each of the other ways. Observe that the next request is $\left(\bigcup_{i=1}^lM'_i\right)\setminus S$ and thus, each of the adversary's type $\kappa$ algorithms serves the request at $0$ cost. However, whenever the online algorithm moves and incurs a cost 1, its configuration becomes equal to that of exactly one of the adversary's type $\kappa$ algorithms. In order to maintain the invariant, this type $\kappa$ algorithm executes a reverse move to attain the configuration $S$, and incurs a cost $1$.

Thus, for the requests which do not result in the beginning of a new phase, the cost incurred by the online algorithm is equal to the total cost incurred by all the adversary's algorithms. Now consider a request which results in a phase change from $\kappa$ to $\kappa'$. The online algorithm incurs a cost $D$. Further, each of the $g(\kappa)+g(\kappa')$ algorithms of the adversary, corresponding to the phase types $\kappa$ and $\kappa'$, has to shift at most $1$ server within every cluster to maintain the invariant. Note that none of the adversary's algorithms ever shifts its servers across clusters. Thus the total cost incurred by the adversary's algorithms is $l\cdot(g(\kappa)+g(\kappa'))\leq l\cdot h(k,l)$. We choose $D>l\cdot h(k,l)$, so that whenever a new phase begins, the adversary pays less than the online algorithm.

Suppose the adversary stops giving requests after $r$ phases, and this results in $m$ requests. The cost incurred by the online algorithm is $rD$ for the requests resulting in a phase change, and $m-r$ for the requests not resulting in a phase change. Hence the cost incurred by the online algorithm over the entire request sequence is $rD+m-r$.

Let us estimate the total cost of the adversary's $h(k,l)$ algorithms. Each of the adversary's algorithms incurs a cost of at most $kD$ to establish the invariant initially. Later, for the $m-r$ requests not resulting in a phase change, the total cost of these algorithms is $m-r$ and for the $r$ requests resulting in a phase change, the total cost of these algorithms is at most $rl\cdot h(k,l)$. Thus the overall cost of all the adversary's algorithms, over all requests, is at most $kD\cdot h(k,l)+rl\cdot h(k,l)+m-r$.

We choose $r>kD\cdot h(k,l)/\left(D-l\cdot h(k,l)\right)$ so that $rD>kD\cdot h(k,l)+rl\cdot h(k,l)$
and thus the cost of the online algorithm is no less than the total cost of the adversary's algorithms. Further, since the adversary runs $h(k,l)$ algorithms, one of these must incur a cost less than $1/h(k,l)$ times the cost of the online algorithm. This proves a lower bound of $h(k,l)$ on the competitive ratio of any online algorithm for $(k,l)$-MSSMS.

Finally, let us calculate $h(k,l)$, the total number of algorithms run by the adversary. This number is given by $\sum\left[\prod_{i=1}^k(k_i+1)-1\right]=\left[\sum\prod_{i=1}^k(k_i+1)\right]-\left[\sum1\right]$, where all the summations are over $(k_1,\ldots,k_l)$ such that each $k_i\geq0$ and $k_1+\cdots+k_l=k$.
Observe that the first term on the right hand side is the number of ways of arranging $2l-1$ identical red balls and $k$ identical green balls on a line, where $k_i$ is the number of green balls between the $(2i-2)^{\text{\tiny{nd}}}$ and the ${2i}^{\text{\tiny{th}}}$ red balls from the left. Thus, this term is equal to ${{k+2l-1}\choose{k}}$. The term $\sum_{(k_1,\ldots,k_l):k_1+\cdots+k_l=k}1$ can be easily seen to be equal to ${k+l-1}\choose{k}$. Thus the number of algorithms run by the adversary is $h(k,l)={{k+2l-1}\choose{k}}-{{k+l-1}\choose{k}}$. \qed
\end{proof}

Now for the second lower bound, we first consider a ``weighted'' variant of MSSMS, where we have some $t\leq k,l$, and $k_1,\ldots,k_t>0$ such that $k_1+\cdots+k_t=k$. The underlying space is the uniform metric space with $k+\lfloor l/t\rfloor$ points. (Henceforward we drop the floor symbols and simply write $l/t$ for $\lfloor l/t\rfloor$.) The servers are put into $t$ groups of size $k_1,\ldots,k_t$ and the servers in the $i^{\text{\tiny{th}}}$ group charge the algorithm at rate $\beta_i$ per unit displacement. Each request is a set of $l/t$ points; specifically, the set of points not occupied by the servers. It can be shown, by a proof analogous to the proof of Theorem 3.2 in \cite{FiatR94}, that there exist $\beta_1,\ldots,\beta_t$ such that any online algorithm for this problem must have competitive ratio at least ${{k+l/t}\choose{k_1,\ldots,k_t,l/t}}/2$. But this problem is essentially $(k,l)$-MSSMS on the metric space consisting of $t$ subspaces separated by infinite distance, where the $i^{\text{\tiny{th}}}$ subspace is the uniform space on $k+l/t$ points scaled by the factor $\beta_i$ and it contains $k_i$ servers, and each request contains $l/t$ points from each of the subspaces. Thus for every $t\leq k,l$, and $k_1,\ldots,k_t>0$ such that $\sum_j k_j=k$, ${{k+l/t}\choose{k_1,\ldots,k_t,l/t}}/2$ is a lower bound on the competitive ratio of any online algorithm for $(k,l)$-MSSMS.

\section{The Offline problem}\label{sec_offline}

We elaborate on the offline $(k,l)$-MSSMS problem in this section. Before that, we briefly describe the offline algorithms for the $k$-server problem and Metrical Service Systems.

The problem of finding the optimal solution to an instance of $k$-server problem can be reduced to the problem of finding a min-cost flow on a suitably constructed directed graph \cite{ChrobakKPV91}, and hence, the offline $k$-server problem can be solved in polynomial time. Note that for any instance of min-cost flow, there exists a solution in which all the flows are integral, and which is no worse than any fractional solution. We will use this fact later. The offline MSS problem can be translated to finding the shortest source-to-sink path in a suitably constructed directed graph of size linear in the size of the instance. Thus, MSS too can be solved in polynomial time. In fact, $(k,l)$-MSSMS can be solved in polynomial time for any constant $k$. However, the problem is NP-hard for any fixed $l\geq2$, when $k$ is allowed to vary.

In the subsequent subsections, we give a natural Integer Linear Program (ILP) for the offline $(k,l)$-MSSMS problem. Although it has an unbounded integrality gap, we use it to design a pseudo-approximation algorithm which uses $kl$ servers, and incurs a cost of at most $l$ times the optimum. We then prove a lower bound which suggests that nothing better than our pseudo-approximation can be achieved in polynomial time. This result has implications in the online setting also. Feuerstein \cite{Feuerstein98} gave a polynomial time online algorithm for MSSMS on uniform spaces which uses $kl$ servers, and achieves a competitive ratio of $kl$ against an adversary using $k$ servers, where $l$, as before, is the size of each request set. Our result implies that it is unlikely that a polynomial time competitive algorithm using less than $kl$ servers exists. 

\subsection{ILP Formulation and Pseudo-approximation algorithm}

For the metric space $(M,d)$, let $S=\{s^1,\ldots,s^k\}$ be the set of initial positions of the $k$ servers and let the request sequence be $\rho=(R_1,\ldots,R_m)$ where $R_i\subseteq M$ and $R_i=\{r_i^1,\ldots,r_i^l\}$. A natural integer linear program is as follows. For each $1\leq i<i'\leq m$ and $1\leq j,j'\leq l$, we have a variable $f(i,j,i',j')$, which is $1$ if some server was present at $r_i^j$ to serve $R_i$, and that server was next made to shift to $r_{i'}^{j'}$ in order to serve $R_{i'}$; and $0$ otherwise. For each $1\leq k'\leq k$, $1\leq i\leq m$, $1\leq j\leq l$, we have a variable $g(k',i,j)$ which is $1$ if the $k'$'th server was first shifted to $r_i^j$, to serve $R_i$, and $0$ otherwise. The objective and the constraints are as given in Figure \ref{fig_ilp}.

\begin{figure}[h]
\begin{eqnarray*}
\text{Minimize} & & \sum_{k',i,j}d(s^{k'},r_i^j)g(k',i,j)+\sum_{i,j,i',j'}d(r_i^j,r_{i'}^{j'})f(i,j,i',j')\\
\sum_{i,j}g(k',i,j) & \leq & 1\text{ for each }1\leq k'\leq k\\
\sum_{k'}g(k',i,j)+\sum_{i'<i,j'}f(i',j',i,j) & \geq & \sum_{i''>i,j''}f(i,j,i'',j'')\text{ for each }1\leq i\leq m\text{, }1\leq j\leq l\\
\sum_{k',j}g(k',i,j)+\sum_{i'<i,j',j}f(i',j',i,j) & \geq & 1\text{ for each }1\leq i\leq m\\
f(i,j,i',j')\in\{0,1\}\text{ for all }i,j,i',j' & & g(k',i,j)\in\{0,1\}\text{ for all }k',i,j
\end{eqnarray*}
\caption{An ILP formulation of the offline $(k,l)$-MSSMS problem}\label{fig_ilp}
\end{figure}

To relax the ILP to an LP, we replace the constraints $f(i,j,i',j')\in\{0,1\}$ and $g(k',i,j)\in\{0,1\}$ by the constraints $0\leq f(i,j,i',j')\leq1$ and $0\leq g(k',i,j)\leq1$ respectively. Unfortunately, this relaxation has an unbounded integrality gap.

\begin{claim}
The LP relaxation of the ILP given in Figure \ref{fig_ilp} has unbounded integrality gap when $k,l>1$.
\end{claim}

\begin{proof}
Consider the uniform metric space on a set of points $M_0\uplus M$, where $M_0$ is the set of $k$ points occupied by the servers initially, and $|M|=kl$.
Let $R_1,\ldots,R_{{kl}\choose{l}}$ be an arbitrary ordering of the $l$-size subsets of $M$. The request sequence is constructed by repeating this ordering $m$ times. This request sequence can be fractionally served as follows. Initially, a $1/l$ fraction of servers shifts to each of the $kl$ points in $M$, and this costs $k$. Having done this, all requests can now be served at zero cost. Thus, the optimum of the LP relaxation is at most $k$. 

However, an integer solution must have cost at least $(k-1)(l-1)$ on the sequence $R_1,\ldots,R_{{kl}\choose{l}}$, starting with any server configuration. This is justified as follows. 
Given an initial server configuration, let $U$ be the set of points visited to serve the sequence $R_1,\ldots,R_{{kl}\choose{l}}$, including the points occupied initially. Then the cost paid is at least $|U|-k$, and $U$ is a hitting set of $R_1,\ldots,R_{{kl}\choose{l}}$. Since the sequence $R_1,\ldots,R_{{kl}\choose{l}}$ contains all $l$-subsets of a $kl$-sized set, any hitting set of this sequence must have size at least $kl-l+1$. Thus, $|U|\geq kl-l+1$, and the cost of serving the sequence is at least $kl-l+1-k=(k-1)(l-1)$.

Therefore, the total cost of the ILP is at least $m(k-1)(l-1)$. Since $m$ can be chosen to be arbitrarily large, the ILP has infinite integrality gap when both $k$ and $l$ are more than $1$. \qed
\end{proof}

We can, however, round the solution of the LP relaxation to get a pseudo-approximation algorithm for MSSMS.

\begin{theorem}
There is a polynomial time algorithm, which computes a feasible solution to a given instance of $(k,l)$-MSSMS using $kl$ servers instead of $k$, and such that the cost of the solution is at most $l$ times the cost of the optimum (fractional) solution of the LP relaxation of the ILP. 
\end{theorem}

\begin{proof}
Given an instance of $(k,l)$-MSSMS with $S$, the set of initial server positions, let 
the vector $(f^*,g^*)$ be $l$ times the optimum (fractional) solution of the LP relaxation. Then we have $\sum_{k',j}g^*(k',i,j)+\sum_{i'<i,j',j}f^*(i',j',i,j)\geq l$ for each $i$.
Thus, for each $i$, there exists some $j_i$ such that $\sum_{k'}g^*(k',i,j_i)+\sum_{i'<i,j'}f^*(i',j',i,j_i)\geq1$. Let $r^*_i=r_i^{j_i}$. The algorithm solves the LP relaxation and finds $r^*_1,\ldots,r^*_m$. It then treats $r^*_1,\ldots,r^*_m$ as an instance of the $kl$-server problem, and assuming $l$ servers to be initially located at each point in $S$, computes the optimum solution. Note that this is a feasible solution to the given instance of $(k,l)$-MSSMS, except that it uses $kl$ servers instead of $k$. To analyze its cost, observe that the vector $(f^*,g^*)$ gives a (fractional) solution to the instance $r^*_1,\ldots,r^*_m$ of the $kl$-server problem. Hence the cost of the optimum solution to this instance is no more then the cost of $(f^*,g^*)$, which is $l$ times the cost of the optimum fractional solution to the LP relaxation. 
Thus, the solution returned by the algorithm 
has cost at most $l$ times that of the optimum of the LP relaxation. \qed
\end{proof}

\subsection{Hardness of Pseudo-Approximation}

The hardness result, that we prove next, essentially implies that nothing better than what the pseudo-approximation algorithm does, can be achieved in polynomial time. The proof involves a reduction from the problem of finding an approximate vertex cover of a given uniform hypergraph. The following results are known about the hardness of approximating hypergraph vertex cover.

\begin{fact}[Khot and Regev, \cite{KhotR08}]\label{fact_hard1}
Assuming the Unique Games Conjecture, (UGC) \cite{Khot_STOC02_UGC}, it is NP-hard to find a factor $l-\epsilon$ approximation of the minimum vertex cover of a given $l$-uniform hypergraph, for all $l\geq2$ and $\varepsilon>0$.
\end{fact}

\begin{fact}[Dinur et. al, \cite{DinurGKR05_VC_hardness}, Dinur and Safra, \cite{DinurS_STOC02_VC_hardness}]\label{fact_hard2}
It is NP-hard to find a factor $l-1-\epsilon$ (resp. $1.36-\epsilon$) approximation of the minimum vertex cover of a given $l$-uniform hypergraph, for all $l\geq3$ (resp. $l=2$) and $\varepsilon>0$.
\end{fact}

\begin{theorem}\label{thm_hardness}
Assuming the Unique Games Conjecture, it is NP-hard to pseudo-approximate the $(k,l)$-MSSMS problem on the uniform metric space with $n$ points, within any factor polynomial in $n$, using at most $k(l-\varepsilon)$ servers, for any fixed $l\geq 2$ and $\varepsilon>0$.
\end{theorem}


\begin{proof}
Suppose there is a polynomial time algorithm $\mathcal{A}$ which, on a metric space with $n$ points, can pseudo-approximate the solution of an instance of $(k,l)$-MSSMS using $K=k(l-\varepsilon)$ servers, within a factor $f(n)=\mathcal{O}(n^p)$ for some constant $p$ independent of $k$. We use this algorithm to construct a polynomial time algorithm $\mathcal{B}$ which, when given a $l$-uniform hypergraph having a vertex cover of size $k$, outputs a vertex cover of the hypergraph of size $K$. By running algorithm $\mathcal{B}$ for all $k$ and picking the best solution, we get a factor $l-\varepsilon$ approximation algorithm for the minimum vertex cover problem on $l$-uniform hypergraphs. This contradicts Fact \ref{fact_hard1}.

Algorithm $\mathcal{B}$ does the following. Suppose it is given a hypergraph $H=(V,E)$ where $|V|=n$, $E=\{E_1,\dots,E_m\}$, $E_i\subseteq V$ $|E_i|=l$, and an integer $k$, with the promise that $H$ has a vertex cover of size $k$. $\mathcal{B}$ takes the uniform metric space on the set $V\uplus W$ where $W=\{w_1,\ldots,w_k\}$. Then it constructs the instance of $(k,l)$-MSSMS, where one server is initially placed at each of the $w_i$'s and the request sequence is $E_1,\ldots,E_m$ repeated more than $kf(n+k)$ (but polynomially many) times. We will call each repetition of the sequence $E_1,\ldots,E_m$ a phase. $\mathcal{B}$ then uses algorithm $\mathcal{A}$ to find a pseudo-approximate solution of this instance, using $K$ servers.

Since the hypergraph has a vertex cover $V'\subseteq V$ of size $k$, the requests can be served by shifting the servers to the points in $V'$ once for all, at a cost $k$. Therefore, the cost of the solution returned by algorithm $\mathcal{A}$ is at most $kf(n+k)$, which is less than the number of phases. Thus, there exists a phase in which $\mathcal{A}$ incurs zero cost, which means that $\mathcal{A}$ does not shift servers during that phase at all. Since all the hyperedges are requested in each phase, the set of points occupied by the $K$ servers in this phase must be a vertex cover of $H$.

Having obtained a solution from algorithm $\mathcal{A}$, $\mathcal{B}$ simply searches for a phase in which $\mathcal{A}$ incurred zero cost, and returns the set of $K$ points occupied by the servers during this phase, as an approximate vertex cover of $H$. It is easy to see that 
$\mathcal{B}$ runs in polynomial time. \qed
\end{proof}

Using the same reduction as above, and Fact \ref{fact_hard2} instead of Fact \ref{fact_hard1}, we obtain the following hardness result, which does not rely on the validity of the UGC.

\begin{theorem}
It is NP-hard to pseudo-approximate the $(k,l)$-MSSMS problem on the uniform metric space with $n$ points, within any factor polynomial in $n$, using at most $k(l-1-\varepsilon)$ (resp. $(1.36-\varepsilon)k$) servers, for any fixed $l\geq 3$ (resp. $l=2$) and $\varepsilon>0$.
\end{theorem}

Theorem \ref{thm_hardness} implies that unless the UGC is false, there does not exist a competitive online algorithm for $(k,l)$-MSSMS, which uses $k(l-\varepsilon)$ servers and runs in polynomial time. This, in particular, implies the optimality of Feuerstein's polynomial time algorithm (Theorem 8 of \cite{Feuerstein98}) using $kl$ servers.

\section{Approaches that fail}\label{sec_limitations}

In this section we discuss a few popular techniques, that produced good algorithms for the $k$-server problem and metrical service systems, but which, in case of MSSMS on arbitrary metric spaces, either do not give competitive online algorithms, or their analysis involves more technical complications.

\subsection{The Harmonic Algorithm}

The Harmonic algorithm is a natural memoryless algorithm and was among the first studied randomized algorithms for the $k$-server problem. In order to serve a request at a point $y$ when the servers are at $x_1,\ldots,x_k$, the algorithm moves the server located at $x_i$ to $y$ with a probability inversely proportional to $d(x_i,y)$. This algorithm has a competitive ratio of $\mathcal{O}(k\cdot2^k)$ \cite{Grove91, BartalG00}. The natural analog of this algorithm for the MSSMS problem would be the algorithm which, when given a request $\{y_1,\ldots,y_l\}$ with the servers at $x_1,\ldots,x_k$, moves the server at $x_i$ to $y_j$ with probability $(d(x_i,y_j))^{-1}/\sum_{i,j}(d(x_i,y_j))^{-1}$.

The Harmonic algorithm, however, does not work even for metrical service system on the line. To see this, consider the metric space of integer points on the line. Let $k=1$, $l=2$ and let the server be initially located at $0$. For $i=1$ to $m$, let the $i$'th request in the sequence be $\{-1,i\}$, and let the $m+1$'st request be $\{-1\}$. The optimal solution is to move the server to $-1$, and this costs $1$. However, it is easy to verify that with probability $1/(i+1)$ the server is at $i$ after the $i$th request. This implies that the interval $[i-1,i]$ is traversed twice with a probability $1/(i+1)$, and not traversed at all with a probability $i/(i+1)$. The interval $[-1,0]$ is traversed once, with probability $1$. The expected cost incurred by the algorithm is thus $1+2\sum_{i=1}^{m}1/(i+1)=2H_m-1$. 

\subsection{The Work Function Algorithm}

The technique of work functions has given the best known deterministic online algorithms, both for the $k$-server problem as well as MSS. It seems to be the most promising algorithm for the MSSMS problem too. The work function algorithm can be described as follows.

Let $(M,d)$ be a metric space. We can think of this metric space as a complete edge-weighted graph on the vertex set $M$, where the weight on the edge $(x,y)$ is $d(x,y)$. A server configuration is a multiset of server positions on $M$. Thus, in case of the $k$-server problem, a server configuration is a multiset of $k$ points from $M$, whereas, in case of the MSS problem, a server configuration is a single point. Given server configurations $X$ and $Y$, let $d(X,Y)$ denote the minimum cost that needs to be paid, in order to change configuration $X$ to $Y$. In other words, $d(X,Y)$ is the weight of the minimum weight matching between the sets $X$ and $Y$. It is easy to see that with this definition, $d$ is a metric on the set of server configurations.

Let $X_0$ denote the initial server configuration and $\rho=(R_1,R_2,\ldots)$ be the sequence of requests in an instance of the problem under consideration ($k$-server, MSS or MSSMS). Let $\rho_i$ denote the prefix of $\rho$ with the first $i$ requests. For any sequence $\rho'$ of requests, the work function $w_{\rho'}$  
is a function on the set of server configurations. The value of this function at a server configuration $X$ is defined to be the minimum cost of starting from $X_0$, serving the request sequence $\rho'$, and ending up with the configuration $X$. Thus, if $w$ is the work function after the entire sequence of requests, then the minimum value of $w(X)$, over all configurations $X$, is the optimum cost of serving the requests.

The work function algorithm (WFA) works as follows. Let $w_i=w_{\rho_i}$. Suppose that the online algorithm has the server configuration $X_{i-1}$ after serving the first $i$ requests. In order to serve the request $R_i$, the algorithm changes the server configuration to $X_i$, which serves $R_i$, and which is such that $w_i(X_i)+d(X_{i-1},X_i)\leq w_i(Y)+d(X_{i-1},Y)$ for every server configuration $Y$. (It can be shown that such an $X_i$ must exist.) This algorithm achieves a competitive ratio of $2k-1$ for the $k$-server problem \cite{KoutsoupiasP95}, while a variant of this algorithm is $\mathcal{O}(l\cdot2^l)$-competitive for MSS with $l$ points per request \cite{Burley96}. It is therefore, natural to expect that WFA for MSSMS will achieve some competitive ratio, which is a function of $k$ and $l$ only. Below, we explain the roadblocks encountered in our attempts to prove this claim along the lines of \cite{Burley96} and \cite{KoutsoupiasP95}.

It can be easily shown that for any work function $w$, and configurations $X$ and $Y$, we have $w(X)\leq w(Y)+d(X,Y)$. A configuration $Y$ is said to \textit{dominate} a configuration $X$ in $w$ if $w(X)=w(Y)+d(X,Y)$. The \textit{support} of a work function is defined to be the set of configurations $X$ such that $X$ is not dominated by any other configuration $Y$, in $w$. It is also easily seen that each such $X$ must serve the last request in the sequence which results in the work function $w$. In case of MSS with $l$ points per request, the size of the support of any work function is at most $l$, and in fact, the support is a subset of the last request. The proof of the competitive ratio of WFA, by Burley \cite{Burley96}, uses this fact crucially. This property, however, does not remain true for MSSMS. The counterexample is as follows. Let $k=l=2$, the metric space be uniform, and the initial configuration be $\{x_0,x'_0\}$. For $i=1$ to $m$, let the $i$'th request be $\{x_i,x'_i\}$, where the points $x_0,x'_0,x_1,x'_1,\ldots,x_m,x'_m$ are all distinct. It is easy to verify that the support of the work function after this sequence contains all $2$-element subsets of $\{x_0,x'_0,x_1,x'_1,\ldots,x_m,x'_m\}$, which contain either $x_m$ or $x'_m$. Thus, the size of support of a work function for MSSMS could depend on the number of points in the metric space, even for fixed $k$ and $l$. It is, therefore, unlikely that the ideas in Burley's proof can be easily extended for MSSMS.

In the case of $k$-server problem, the proof of the competitive ratio of WFA, by Koutsoupias and Papadimitriou \cite{KoutsoupiasP95}, depends on a property of work functions, which they call \textit{quasi-convexity}. The quasi-convexity property states that for any work function $w_{\rho}$ arising from an instance $\rho$ of the $k$-server problem, and any two configurations $X$ and $Y$, there exists a bijection $g:X\longrightarrow Y$, such that for any $X'\subseteq X$, $w(g(X')\cup(X\setminus X'))+w(X'\cup g(X\setminus X'))\leq w(X)+w(Y)$. This property is also not true for work functions in case of MSSMS. As a counterexample, consider the uniform metric space. Let the initial configuration be $\{1,2\}$ and let the request sequence be $\{3,4\},\{5,6\},\{3,6\},\{5,4\}$. If $w$ is the work function after this sequence, then $w(\{3,5\})=w(\{4,6\})=2$ and the value of $w$ at any other configuration is more than $2$. Thus $w$ is not quasi-convex. It would, therefore, be interesting to investigate whether there exists an analog of quasi-convexity for $l>1$.

\section{Open Problems}\label{sec_open_problems}

We conclude with a number of interesting problems left open. The most important problem among these is the following.
\begin{problem}
Design an $f(k,l)$-competitive deterministic / randomized algorithm for $(k,l)$-MSSMS on arbitrary metric spaces, for some function $f$.
\end{problem}

For $(k,l)$-MSSMS on arbitrary metric spaces, we do not have an online algorithm with competitive ratio determined by $k$ and $l$ alone. We believe that such an algorithm exists, and the Work Function Algorithm is the prime candidate for the deterministic case. However, it is not known whether the Work Function Algorithm is competitive, even for the Generalized $3$-server problem. We know that if such a deterministic algorithm exists, it must perform a super-polynomial amount of computation on the input. It would be interesting even to find a constant factor competitive algorithm for $(2,2)$-MSSMS on special metric spaces, such as the line.

The lower bound obtained in Section \ref{sec_general_lb} on the competitive ratio of such a deterministic algorithm is polynomial in $l$ for a fixed $k$, and vice versa. We believe that a better lower bound exists, which is exponential in $l$ for a fixed $k$, and at least polynomial in $k$ for a fixed $l$. Such a lower bound can possibly proved by a construction which combines our ideas and the ideas in the $\Omega(2^l)$ lower bound proof for MSS by Fiat et. al. \cite{FiatFKRRV98}. In the randomized case, we believe that a lower bound which is polynomial in $l$ and logarithmic in $k$ should hold.

\begin{problem}
Prove an exponential (resp $\Omega(\poly(l)\cdot\log k)$) lower bound on the competitive ratio of any deterministic (resp. randomized) algorithm for $(k,l)$-MSSMS.
\end{problem}

For deterministic algorithms on uniform metric spaces, our upper and lower bounds on the competitive ratio differ by the factor $k$. We believe that the upper bound can be improved by carefully choosing the hitting sets, in the Hitting Set algorithm, described in Section \ref{sec_hitting_set}. Similarly, there is a significant gap between the 
bounds in case of randomized algorithms. In the upper bound (resp. lower bound) results, we have made the conservative assumption that whenever the algorithm (resp. adversary) faults, it can potentially shift all the $k$ servers, and hence the cost incurred could be $k$, while the adversary (resp. algorithm) shifts at most one server per fault. This assumption introduces a gap of a factor of $k$ between the bounds. 

\begin{problem}
For $(k,l)$-MSSMS on uniform metric spaces, close the gap between the lower bound of ${{k+l}\choose{l}}-1$ and the upper bound of $k\cdot\left({{k+l}\choose{l}}-1\right)$ in the deterministic case, and the lower bound of $\Omega(\log kl)$ and the upper bound of $\mathcal{O}(k^3\log l)$ in the randomized case.
\end{problem}

\bibliographystyle{plain}
\bibliography{references.bib,references_misc.bib,references_hardness.bib,references_generalized_kserver.bib}

\begin{thebibliography}{10}

\bibitem{BansalBMN11}
Nikhil Bansal, Niv Buchbinder, Aleksander Madry, and Joseph Naor.
\newblock A polylogarithmic-competitive algorithm for the k-server problem.
\newblock In {\em IEEE 52nd Annual Symposium on Foundations of Computer
  Science}, pages 267--276. IEEE, 2011.

\bibitem{BansalBN07}
Nikhil Bansal, Niv Buchbinder, and Joseph Naor.
\newblock A primal-dual randomized algorithm for weighted paging.
\newblock In {\em 48th Annual IEEE Symposium on Foundations of Computer
  Science}, pages 507--517. IEEE Computer Society, 2007.

\bibitem{BartalBM01}
Yair Bartal, B{\'e}la Bollob{\'a}s, and Manor Mendel.
\newblock A ramsey-type theorem for metric spaces and its applications for
  metrical task systems and related problems.
\newblock In {\em 42nd Annual Symposium on Foundations of Computer Science},
  pages 396--405. IEEE Computer Society, 2001.

\bibitem{BartalG00}
Yair Bartal and Eddie Grove.
\newblock The harmonic $k$-server algorithm is competitive.
\newblock {\em Journal of the ACM}, 47(1):1--15, 2000.

\bibitem{BorodinE}
Allan Borodin and Ran El-Yaniv.
\newblock {\em Online computation and competitive analysis}.
\newblock Cambridge University Press, 1998.

\bibitem{BorodinE99}
Allan Borodin and Ran El-Yaniv.
\newblock On randomization in on-line computation.
\newblock {\em Information and Computation}, 150(2):244--267, 1999.

\bibitem{Burley96}
William~R. Burley.
\newblock Traversing layered graphs using the work function algorithm.
\newblock {\em Journal of Algorithms}, 20(3):479--511, 1996.

\bibitem{ChrobakKPV91}
Marek Chrobak, Howard~J. Karloff, T.~H. Payne, and Sundar Vishwanathan.
\newblock New results on server problems.
\newblock {\em SIAM Journal on Discrete Mathematics}, 4(2):172--181, 1991.

\bibitem{ChrobakL92}
Marek Chrobak and Lawrence Larmore.
\newblock The server problem and on-line games.
\newblock In {\em On-Line Algorithms: Proceedings of a DIMACS Workshop},
  volume~7 of {\em DIMACS Series in Discrete Mathematics and Theoretical
  Computer Science}, pages 11--64, 1992.

\bibitem{ChrobakL91}
Marek Chrobak and Lawrence~L. Larmore.
\newblock An optimal on-line algorithm for k-servers on trees.
\newblock {\em SIAM Journal on Computing}, 20(1):144--148, 1991.

\bibitem{ChrobakL93}
Marek Chrobak and Lawrence~L. Larmore.
\newblock Metrical service systems: Deterministic strategies.
\newblock Technical report, 1993.

\bibitem{ChrobakS04}
Marek Chrobak and Ji\v{r}{\'i} Sgall.
\newblock The weighted 2-server problem.
\newblock {\em Theoretical Computer Science}, 324(2-3):289--312, 2004.

\bibitem{DinurGKR05_VC_hardness}
Irit Dinur, Venkatesan Guruswami, Subhash Khot, and Oded Regev.
\newblock A new multilayered {PCP} and the hardness of hypergraph vertex cover.
\newblock {\em SIAM Journal on Computing}, 34(5):1129--1146, 2005.

\bibitem{DinurS_STOC02_VC_hardness}
Irit Dinur and Shmuel Safra.
\newblock The importance of being biased.
\newblock In {\em Proceedings on 34th Annual ACM Symposium on Theory of
  Computing}, pages 33--42. ACM, 2002.

\bibitem{Feuerstein98}
Esteban Feuerstein.
\newblock Uniform service systems with $k$ servers.
\newblock In {\em LATIN '98: Theoretical Informatics, Third Latin American
  Symposium}, volume 1380 of {\em Lecture Notes in Computer Science}, pages
  23--32. Springer, 1998.

\bibitem{FiatFKRRV98}
Amos Fiat, Dean~P. Foster, Howard~J. Karloff, Yuval Rabani, Yiftach Ravid, and
  Sundar Vishwanathan.
\newblock Competitive algorithms for layered graph traversal.
\newblock {\em SIAM Journal on Computing}, 28(2):447--462, 1998.

\bibitem{FiatRR90}
Amos Fiat, Yuval Rabani, and Yiftach Ravid.
\newblock Competitive k-server algorithms (extended abstract).
\newblock In {\em 31st Annual Symposium on Foundations of Computer Science},
  pages 454--463. IEEE Computer Society, 1990.

\bibitem{FiatR94}
Amos Fiat and Moty Ricklin.
\newblock Competitive algorithms for the weighted server problem.
\newblock {\em Theoretical Computer Science}, 130(1):85--99, 1994.

\bibitem{Grove91}
Edward~F. Grove.
\newblock The harmonic online k-server algorithm is competitive.
\newblock In {\em Proceedings of the 23rd Annual ACM Symposium on Theory of
  Computing}, pages 260--266. ACM, 1991.

\bibitem{KarlinMMO94}
Anna~R. Karlin, Mark~S. Manasse, Lyle~A. McGeoch, and Susan~S. Owicki.
\newblock Competitive randomized algorithms for nonuniform problems.
\newblock {\em Algorithmica}, 11(6):542--571, 1994.

\bibitem{Khot_STOC02_UGC}
Subhash Khot.
\newblock On the power of unique 2-prover 1-round games.
\newblock In {\em Proceedings on 34th Annual ACM Symposium on Theory of
  Computing}, pages 767--775. ACM, 2002.

\bibitem{KhotR08}
Subhash Khot and Oded Regev.
\newblock Vertex cover might be hard to approximate to within $2-\varepsilon$.
\newblock {\em Journal of Computer and System Sciences}, 74(3):335--349, 2008.

\bibitem{Koutsoupias09}
Elias Koutsoupias.
\newblock The $k$-server problem.
\newblock {\em Computer Science Review}, 3(2):105--118, 2009.

\bibitem{KoutsoupiasP95}
Elias Koutsoupias and Christos~H. Papadimitriou.
\newblock On the k-server conjecture.
\newblock {\em Journal of the ACM}, 42(5):971--983, 1995.

\bibitem{Lovasz_skewBollobas}
L{\'a}szl{\'o} Lov{\'a}sz.
\newblock Flats in matroids and geometric graphs.
\newblock In {\em Combinatorial surveys ({P}roc. {S}ixth {B}ritish
  {C}ombinatorial {C}onf., {R}oyal {H}olloway {C}oll., {E}gham, 1977)}, pages
  45--86. Academic Press, London, 1977.

\bibitem{ManasseMS88}
Mark~S. Manasse, Lyle~A. McGeoch, and Daniel~Dominic Sleator.
\newblock Competitive algorithms for on-line problems.
\newblock In {\em Proceedings of the 20th Annual ACM Symposium on Theory of
  Computing}, pages 322--333. ACM, 1988.

\bibitem{PapadimitriouY91}
Christos~H. Papadimitriou and Mihalis Yannakakis.
\newblock Shortest paths without a map.
\newblock {\em Theoretical Computer Science}, 84(1):127--150, 1991.

\bibitem{Ramesh95}
H.~Ramesh.
\newblock On traversing layered graphs on-line.
\newblock {\em Journal of Algorithms}, 18(3):480--512, 1995.

\bibitem{Sitters11}
Ren{\'e} Sitters.
\newblock The generalized work function algorithm is competitive for the
  generalized 2-server problem.
\newblock {\em CoRR}, abs/1110.6600, 2011.

\bibitem{StougieV02}
Leen Stougie and Arjen P.~A. Vestjens.
\newblock Randomized algorithms for on-line scheduling problems: how low can't
  you go?
\newblock {\em Operations Research Letters}, 30(2):89--96, 2002.

\bibitem{Yao77}
Andrew Chi-Chih Yao.
\newblock Probabilistic computations: Toward a unified measure of complexity
  (extended abstract).
\newblock In {\em 18th Annual Symposium on Foundations of Computer Science},
  pages 222--227. IEEE Computer Society, 1977.

\end{thebibliography}
\end{document}